\documentclass{article} 
\usepackage[utf8]{inputenc}
\usepackage[T1]{fontenc}
\usepackage[unicode]{hyperref}
\usepackage{xspace}
\usepackage{booktabs}
\usepackage{amsmath}
\usepackage{named}
\usepackage{amssymb}
\usepackage{amsthm}
\usepackage{times}

\newcommand*{\csf}{\ensuremath{{s}}}
\newcommand*{\ssf}{\ensuremath{{s}}}
\newcommand*{\Sinz}{\ensuremath{\Phi}}
\newcommand*{\Eq}{\ensuremath{\Psi}}

\newcommand*{\vars}{\ensuremath{\mathrm{vars}}}

\newcommand*{\CBound}{\ensuremath{\mathsf{CBound}}}
\newcommand*{\edp}{\ensuremath{\mathsf{edp}}}
\newcommand*{\natP}{\ensuremath{\mathbb{N}^{+}}}
\newcommand*{\Mult}{\ensuremath{\mathsf{edp_{m}}}}
\newcommand*{\CMult}{\ensuremath{\mathsf{edp_{cm}}}}
\newcommand*{\cmult}{\ensuremath{\mathsf{cmult}}}
\newcommand*{\Treengeling}{\texttt{Treengeling}\xspace}
\newcommand*{\Lingeling}{\texttt{Lingeling}\xspace}
\newcommand*{\Glucose}{\texttt{Glucose}\xspace}
\newcommand*{\druptrim}{\texttt{drup-trim}\xspace}

\renewcommand*{\mu}{\ensuremath{\textunderscore}}

\newtheorem{theorem}{Theorem}
\newtheorem{example}[theorem]{Example}
\newtheorem{proposition}[theorem]{Proposition}

\begin{document}

\title{Computer-Aided Proof of Erd\H{o}s Discrepancy Properties}
\author{Boris Konev    and
Alexei Lisitsa \\{Liverpool University}
}

\date{}
\maketitle

\begin{abstract}
In 1930s Paul Erd\H{o}s conjectured that for any positive integer $C$ in any
infinite $\pm 1$ sequence $(x_n)$ there exists a subsequence $x_d, x_{2d},
x_{3d},\dots, x_{kd}$, for some positive integers $k$ and $d$, such that
$\mid \sum_{i=1}^k x_{i\cdot d} \mid >C$.  
The conjecture has been referred to  as one of the major open problems in 
combinatorial number theory and discrepancy theory.  
For the particular case of $C=1$ a human proof of the conjecture exists; for
$C=2$ a bespoke computer program had generated sequences of length $1124$
of discrepancy $2$,
but the status of the conjecture remained open even for such a small bound.
We show that by encoding the problem into Boolean satisfiability and applying the 
state of the art SAT solvers, one can obtain a discrepancy $2$ sequence of length $1160$
and a \emph{proof} of the Erd\H{o}s discrepancy conjecture for
$C=2$, claiming that no discrepancy 2 sequence of length $1161$, or more, exists.    
In the similar way, we obtain a precise bound of $127\,645$ on the maximal
lengths of both multiplicative and completely multiplicative sequences of discrepancy
$3$.
\end{abstract}
\section{Introduction}
Discrepancy theory is a branch of mathematics dealing with 
irregularities of distributions of points in some space in combinatorial, measure-theoretic
and geometric settings~\cite{Beck,Chazelle,Matusek,BeckSos}. 
The paradigmatic combinatorial discrepancy theory  problem can be described in
terms of a hypergraph $\mathcal{H} = (U,S)$, that is, a set $U$  and a family
of its subsets $S \subseteq 2^{U}$. 
Consider a colouring $c : U \rightarrow \{+1,-1\}$ of the elements of $U$ in  
\emph{blue} $(+1)$ and \emph{red}  ($-1$) colours. Then one may ask whether there 
exists a colouring of the elements of ${U}$ such that colours are
distributed uniformly in every element of $S$ or a discrepancy of colours is always inevitable.
Formally, the discrepancy (deviation from a uniform distribution) of a hypergraph $\mathcal{H}$ is defined as
$\min_{c} (\max_{s \in S} \,| \sum_{e \in s} c(e) |\,)$.
Discrepancy theory has found applications 
in {computational complexity}~\cite{Chazelle},
complexity of communication~\cite{DBLP:journals/dam/Alon92}
and {differential privacy}~\cite{Muthukrishnan:2012:OPH:2213977.2214090}. 

One of the oldest problems of discrepancy theory is the 
discrepancy of hypergraphs over sets of natural numbers with the  subsets
(hyperedges) forming arithmetical progressions over these sets~\cite{MS96}. 
 Roth's theorem~\cite{Roth64}, one of the main results in the area, states that
for the hypergraph formed by the arithmetic progressions in $\{1,\dots, n\}$, that is
$\mathcal{H}_n=(U_n, S_n)$, where $U_n = \{1,2, \ldots, n\}$ and
elements of $S_n$ being of the form $(ai+b)$ for arbitrary $a, b$,
the discrepancy grows at least as $\frac{1}{20}n^{1/4}$.

\vspace*{1pt}

Surprisingly, for the more restricted case of \emph{homogeneous} arithmetic progressions of the form
$(ai)$, the question of the discrepancy bounds is open for more than eighty years.   
In 1930s Paul Erd\H{o}s conjectured~\cite{Unsolved} that
discrepancy is unbounded. Independently the same conjecture has been raised by 
\cite{chudakov}.
Proving or  disproving this conjecture became one of the major open problems in
combinatorial number theory and discrepancy theory. It has been referred to as the
\emph{Erd\H{o}s discrepancy problem} (EDP)~\cite{Beck,BeckSos,NiTa}.

The  expected value of the discrepancy of 
random  $\pm 1$ sequences of length $n$ grows as
$n^{1/2+o(1)}$ and the explicit constructions of a sequence with slowly growing
discrepancy at the rate of $\log_3 n$ have been demonstrated~\cite{gowers,BCC10}. 
By considering cases, one can see  that
any $\pm1$ sequence containing $12$ or more elements has discrepancy at least $2$;
that
is, Erd\H{o}s's conjecture holds for the particular case $C=1$ 
(also implied by a stronger result of Mathias~\shortcite{Mathias}). 
Until recently the status of the conjecture remained unknown for all other
values of $C$.  Although widely believed not to be the case, there was still a
possibility that an infinite sequence of discrepancy 2 existed. 

The conjecture whether the discrepancy of an arbitrary $\pm1$ sequence is
unbounded is equivalent to the question whether the discrepancy of a completely
multiplicative $\pm1$ sequence is unbounded, where a sequence is completely
multiplicative if $x_{m\cdot n} = x_m\cdot x_n$ for any $m,n$~\cite{Unsolved}.
For completely multiplicative sequences the choices of how the sequence can be
constructed are severely limited as the entire sequence is defined by the values of
$x_i$ for prime $i$. The longest completely multiplicative sequence of
discrepancy $2$ has length $246$~\cite{Polymath2}.  For discrepancy $3$ the
bound was not known.

The EDP has attracted renewed interest in 2009-2010  as it became a topic of the fifth
Poly\-math 
project~\shortcite{Polymath}
a widely publicised endeavour in
collective math initiated by T.~Gowers~\shortcite{gowersblog}.   As part of this
activity  an attempt has been made
to attack the problem using computers (see discussion in \cite{Polymath}). A
purposely written computer program had successfully found $\pm1$ sequences of
length $1124$ and discrepancy $2$; however, no further progress has been made
leading to a claim ``given how long a finite sequence can
be, it seems unlikely that we could answer this question just by a clever
search of all possibilities on a computer''~\cite{Polymath}.

The status of the Erd\H{o}s discrepancy conjecture for $C=2$ has been
settled by the authors of this article 
\cite{KLArx14}, \cite{KLSAT14} by reduction to SAT.  The method is
based on establishing the correspondence between $\pm1$ sequences that violate
a given discrepancy bound and words accepted by of a finite automaton.  Traces
of this automaton are represented then by a propositional formula and state
of the art SAT solvers are used to prove that the longest $\pm1$ sequence of
discrepancy $2$ contains $1160$ elements. 
A $13\,900$ long $\pm1$ sequence of discrepancy $3$ was also constructed.

This article is a revised and extended version of~\cite{KLSAT14}.  We use a
different smaller SAT encoding of the Erd\H{o}s discrepancy problem, which is
based on the sequential counter encoding of the \emph{at
most} cardinality constraints\footnote{We are grateful to Donald E. Knuth for
pointing us in that direction.}. The impact of the new encoding is twofold.
Firstly, it allows us to significantly reduce the size of the machine-generated
proof of the fact that any sequence longer than $1160$  has discrepancy at
least $3$.
Secondly, by combining the new encoding with additional restrictions that the
sequence is multiplicative, or completely multiplicative, we improve
significantly the lower bound on the length of sequences of discrepancy $3$. We
prove the surprising result that $127\,645$, the length  of the longest completely multiplicative
sequence of discrepancy $3$, is also the maximal length of a 
multiplicative sequence of discrepancy $3$, which is not the case for $C=1$ and
$C=2$.  The article also contains
detailed argumentation, examples and complete proofs.

The article is organised as follows. In Section~\ref{sec:preliminaries} we
introduce the main terms and definition. In Section~\ref{sec:encoding} we
describe the new SAT encoding of the Erd\H{o}s discrepancy problem.
Results and conclusions are discussed in Sections ~\ref{sec:experiments}
and~\ref{sec:conclusion} respectively.  To improve readability a number of
technical proofs have been deferred to an appendix.

\section{Preliminaries}\label{sec:preliminaries}
We divide this section into three parts: main definitions for the Erd\H{o}s 
discrepancy problem, some background and definitions for SAT solving, and
sequential counter-based SAT encoding of cardinality constraints.

Since number $1$ is used both as an element of $\pm1$ sequences and 
as the logical value \emph{true}, to avoid confusion,  in what follows we write
$1$ to refer to the logical value \emph{true} and $+1$ to refer to elements of
$\pm1$ sequences.  We also use the following naming convention: we write
$x_1,\dots x_n$ for $\pm1$ sequences, $p_1,\dots, p_n$ for sequences of
propositions, and $a_1,\dots, a_n$ for $0/1$ Boolean sequences.

\subsection{Discrepancy of $\pm$1 Sequences}
A $\pm1$ sequence of length $n$ is a function $\{1,\dots,n\}\to\{-1,+1\}$.  An
infinite $\pm1$ sequence is a function $\natP\to\{1,-1\}$, where $\natP$ is the
set of positive natural numbers.  We write $x_1,\dots, x_n$ to denote a finite
$\pm1$ sequence of length $n$, and $(x_n)$ to denote an infinite sequence. We
refer to  the $i$-th element of a sequence $x$, that is the value of $x(i)$,
as $x_i$. 
A (finite or infinite) $\pm1$ sequence $x$ is \emph{completely multiplicative} \cite{apostol} if
\begin{equation}\label{eq:multdef1}
        x_{m\cdot n} = x_m\cdot x_n,\; \textrm{for all $m,n\in \natP$.}
\end{equation}
The sequence is \emph{multiplicative} if (\ref{eq:multdef1}) is only required
for 
coprime $m$ and $n$.

It is easy to see that a sequence $x$ is
completely multiplicative if, and only if, $x_1 = +1$ and for
the canonical representation  $m = \prod_{i=1}^{k} {p_i}^{\alpha_i}$, where $p_1<
p_2<\dots< p_k$ are primes and $\alpha_i\in\natP$, 
we have $x_{m} = \prod_{i=1}^k (x_{p_i})^{\alpha_i}$.
This observation leads to a more computationally friendly definition of
completely multiplicative sequences: $x$ is completely multiplicative if,
and only if,
\begin{equation}\label{eq:multdef2}
\begin{minipage}{0.85\textwidth}$x_1 = +1$ and for every composite $m$ we have $x_m = x_i\cdot x_j$, for \emph{some}
$i$ and $j$, non-trivial divisors of $m$.
\end{minipage}
\end{equation}

\smallskip

The Erd\H{o}s discrepancy problem can be naturally described in terms of $\pm1$
sequences (and this is how Erd\H{o}s himself introduced it~\shortcite{Unsolved}).
Erd\H{o}s's conjecture states that for any $C>0$ in any infinite $\pm 1$
sequence $(x_n)$ there exists a subsequence $x_d, x_{2d}, x_{3d},\dots,
x_{kd}$, for some positive integers $k$ and $d$, such that $| \sum_{i=1}^k
x_{i\cdot d} \,| >C$. 

The general definition of discrepancy given above can be specialised in this
case as follows. The discrepancy of a finite  $\pm 1$ sequence
$ x_1,\dots,x_n$ of length $n$ can be defined as $\max_{d=1 ,\ldots, n} (|
\sum_{i=1}^{\lfloor \frac{n}{d} \rfloor } x_{i\cdot d} |)$. For an infinite
sequence $(x_n)$ its discrepancy is the supremum of discrepancies of all its
initial finite fragments.

\begin{example}\label{ex:edp1}
It is easy to see why any $\pm1$ sequence containing $12$ elements has
discrepancy at least $2$. For the proof of contradiction, suppose that 
the discrepancy of some $\pm1$ sequence $x_1,\dots, x_{12}$ is $1$. 
Assume that $x_1$ is $+1$.  
We write 
$$ 
   (+1,\mu,\mu,\mu,\mu,\mu,\mu,\mu,\mu,\mu,\mu,\mu)
$$
to track progress in this example, that is,  we put specific values 
$+1$ or $-1$ into positions $i: 1\leq i\leq 12$, to indicate decisions on $x_i$
which have been taken so far, and mark positions $x_i$, for which no decision has 
been made by an underscore.

Notice that $x_2$ must be $-1$ for otherwise $x_1+x_2 = 2$. So, we progress to
$$ 
   (+1,-1,\mu,\mu,\mu,\mu,\mu,\mu,\mu,\mu,\mu,\mu).
$$
Then the $4$th element of the sequence must be $+1$ for otherwise 
for $d=2$ the sum $x_d + x_{2d} = x_2+x_4 = -2$. So we progress to
$$ 
   (+1,-1,\mu, +1,\mu,\mu,\mu,\mu,\mu,\mu,\mu,\mu).
$$
Then the $3$rd element of the sequence must be $-1$ for otherwise $x_1+\dots+x_4 = 2$ and
so we come to
$$ 
   (+1,-1, -1, +1,\mu,\mu,\mu,\mu,\mu,\mu,\mu,\mu).
$$
Repeating the reasoning above for 
$x_3$ and $x_6$ followed by $x_5$,
for $x_4$ and $x_8$ followed by $x_7$, 
for $x_5$ and $x_{10}$ followed by $x_9$
and finally for $x_6$ and $x_{12}$ followed by $x_{11}$ we progress to
\begin{equation}\label{eq:example1}
   (+1,-1, -1, +1,-1,+1,+1,-1,-1,+1,+1,-1).
\end{equation}
But then for $d=3$ we have $x_{d}+x_{2d} + x_{3d} + x_{4d} = x_3 + x_6+ x_9 + x_{12} = -2$. 
So we derive a contradiction.
It can be checked  in a similar way that the other possibility of $x_1$ being
$-1$ also leads to a contradiction.

The first eleven elements of the sequence (\ref{eq:example1}) form a discrepancy $1$ sequence.
It is multiplicative but not completely multiplicative as $x_9$  is $-1$.
Reasoning similar to the one above shows that there
exists a unique longest completely multiplicative $\pm1$ sequence of
discrepancy $1$ which has nine elements:
$$
   (+1, -1, -1, +1, -1, +1, +1, -1, +1).
$$
\end{example}

\subsection{Propositional Satisfiability Problem}
We assume standard definitions for propositional logic (see, for example,
\cite{rautenberg}).  Propositional formulae are defined over Boolean constants
\emph{true} and \emph{false}, denoted by $1$ and $0$, respectively, and the set
of Boolean variables (or propositions) $PV$ as follows: Boolean constants $0$
and $1$, as well as the elements of $PV$, are formulae; if $\Phi$ and $\Psi$
are formulae then so are $\Phi\land\Psi$ (conjunction), $\Phi\lor\Psi$
(disjunction), $\Phi\rightarrow\Psi$ (implication), $\Phi\leftrightarrow\Psi$
(equivalence) and $\lnot \Phi$ (negation).  We
typically use letters $p$, $q$ and $s$ to denote propositions and 
capital Greek letters $\Phi$ and $\Psi$ to denote propositional formulae. Whenever
necessary, subscripts and superscripts are used. We use $\vars(\Phi)$ to denote
the set of all propositions occurring in the formula $\Phi$.

Every propositional formula can be reduced to conjunctive normal form.  Propositions and
negations of propositions are called \emph{literals}.
When the negation is applied to a literal, double negations are implicitly 
removed, that is, if $l$ is $\lnot p$ then $\lnot l$ is $p$.
A disjunction of literals
is called a \emph{clause}. A clause containing exactly one literal is called
a unit clause. A conjunction of clauses is called a
propositional formula in \emph{conjunctive normal form}, a CNF formula for short. 
A clause can be represented by the set of its literals and the empty 
clause correspond to $0$ (\emph{false}).
A CNF formula can be represented by the set its clauses. These
representations are used interchangeably. 
We typically use meaningful terms typeset in sans serif font,
for example $\edp$ or $\cmult$,
to highlight the fact that a propositional formula is a CNF formula of interest.

For a propositional formula $\Phi$, we write $\Phi(p_1,\dots,p_n)$ to indicate
that $\{p_1,\dots,p_n\}\subseteq \vars(\Phi)$. Propositions $p_1,\dots, p_n$
are designated as `input' propositions in this case, and the intended meaning
is that formula $\Phi$ encodes some property of $p_1,\dots, p_n$. Then the
expression $\Phi(q_1,\dots,q_n)$ denotes the result of simultaneous replacement
of every occurrence of $p_i$ in $\Phi$  with $q_i$, for $1\leq i \leq n$.

The semantics of propositional formulae is given by interpretations (also termed
assignments).  An
interpretation $I$ is a mapping $PV\to\{0,1\}$ extended to literals, clauses,
CNF formulae and propositional formulae in general in the usual way.
For an assignment $I$ and a formula $\Phi$ we say that $I$ satisfies $\Phi$ 
(or $I$ is a model of $\Phi$) if $I(\Phi) = 1$.
A formula $\Phi$ is satisfiable if there exists an assignment that satisfies
it, and unsatisfiable otherwise.

Despite the non-tractability of the satisfiability problem, the tremendous
progress in recent years made it possible to solve many interesting hard
problems by first expressing them as a propositional formula and then
using a {SAT solver} for obtaining a solution~\cite{satHB}. In addition to
returning a satisfying assignment, if the input formula is satisfiable, some
SAT solvers are also capable to return a proof (or certificate) of
unsatisfiability. 

Reverse Unit Propagation (RUP) proofs constitute a compact representation of 
the resolution refutation of the given formula~\cite{GoldbergN03} in the following sense.
\emph{Unit propagation} is a CNF formula transformation
technique, which simplifies the formula by fixing the values of 
propositions occurring to its unit clauses to satisfy these clauses. 
That is, if the unit clause $(p)$ occurs in the CNF formula then 
all occurrences of $p$ are replaced by $1$ and if the unit clause
$(\lnot p)$  occurs in the CNF formula, all occurrences of 
$p$ are replaced by $0$. Then the CNF formula is simplified in the obvious way.
A clause $C = (l_1,\dots l_m)$ is a RUP inference from the input CNF formula
$\Psi$ if adding the unit clauses $(\lnot l_1),\dots,(\lnot l_m)$ to $\Psi$ makes the whole
formula refutable by unit propagation.  A RUP unsatisfiability certificate is
the sequence of clauses $C_1,\dots C_m$ such that for every $1\leq i \leq m$ 
the clause $C_i$ is a RUP inference from $\Psi\cup\{C_1,\dots,C_{i-1}\}$ and
$C_m$ is the empty clause.  Every unsatisfiable CNF formula has a RUP
unsatisfiability certificate~\cite{GoldbergN03}.

Delete Reverse Unit Propagation (DRUP) proofs extend  RUP proofs by including
extra information about the proof search process, namely clauses that have been
discarded by the solver.  Eliminating this extra information from a DRUP proof
converts it to a valid RUP proof.  DRUP proofs are somewhat longer but they are
significantly faster to verify than a RUP proof~\cite{druptrim}.

\subsection{Sequential Counter-Based SAT Encoding of Cardinality Constraints}\label{sec:sinz_card}
Cardinality constraints~\cite{cardconstrHB} are expressions that impose restrictions on
interpretations by specifying numerical bounds on the number of propositions,
from a fixed set of propositions, that can be assigned value $1$. The \emph{at
most $r$} constraint over the set of propositions $\{p_1,\dots, p_n\}$, written
as $p_1+\dots+p_n\leq r$, holds for an interpretation $I$ if, and only if,
at most $r$ propositions among $p_1,\dots, p_n$ are true under $I$.

A SAT encoding for cardinality constraints of the
form $p_1+\dots + p_n\leq r$ based on a sequential counter circuit has been
suggested by Sinz~\shortcite{Sinz05}. In this encoding, auxiliary propositions
$\csf^k_j$ are introduced to represent a unary counter storing the partial sums
of prefixes of $p_1,\dots, p_n$ so that whenever $\sum_{i=1}^j p_i\geq k$, for
some $j\leq n$, we have $\csf^k_j=1$.

We slightly simplify the presentation of~\cite{Sinz05} as follows.
Let formula $\Sinz(p_1,\dots,p_n)$
be the conjunction of 
\begin{eqnarray}
\label{eq:sinz_card_1}%
  \csf^k_j \leftrightarrow (\csf^k_{j-1}\lor (\csf^{k-1}_{j-1}\land p_j)), &\quad & \textrm{for all $1\leq k\leq n$, $1\leq j\leq n$};\\
\label{eq:sinz_card_2}%
  (\lnot \csf^k_j), & &\textrm{for all $0\leq j<k \leq n$};\\
\label{eq:sinz_card_3}%
  (\csf^k_j), & &\textrm{for $k=0$ and all $0\leq j \leq n$}.
\end{eqnarray}
Recall that we write $\Sinz(p_1,\dots,p_n)$ to highlight the fact that 
$p_1,\dots, p_n$ are designated `input' propositions; the set of 
all propositions 
of $\Sinz(p_1,\dots,p_n)$ is
$\vars(\Sinz(p_1,\dots,p_n)) = \{p_1,\dots,p_n\}\cup\bigcup_{k=1}^n\bigcup_{j=1}^n\{\csf^k_j\}$.

Notice that instead of including formulae (\ref{eq:sinz_card_2}) and
(\ref{eq:sinz_card_3}) explicitly in the encoding, one can directly
modify (\ref{eq:sinz_card_1}) by replacing  
all occurrences of $\csf^k_j$, for $0\leq j<k \leq n$, with $0$ (the truth
value \emph{false}) and all occurrences of $\csf^k_j$, for $k=0$ and all 
$0\leq j \leq n$, with $1$ (the truth value \emph{true}).
Then, for example, for $k=j=1$ formula (\ref{eq:sinz_card_1}) simplifies to
$\csf^1_1\leftrightarrow p_1$. We write (\ref{eq:sinz_card_2}) and
(\ref{eq:sinz_card_3})  explicitly for the exposition purposes.

The proof of the following statement can be extracted from~\cite{Sinz05}. 
It is based on the observation that the sum of the first $j$ elements of the $0/1$ sequence
$p_1,\dots, p_n$ exceeds $k$ if, and only if, either the sum of the first 
$j-1$ elements already exceeds $k$, or the sum of the first $j-1$ elements is $k$ and 
the $j$-th element of the sequence is $1$.
We give the formal proof in an appendix for completeness of the presentation.
\begin{proposition}\label{prop:sinz}
Let $\Sinz(p_1,\dots,p_n)$ be as defined above.  Then 
\begin{enumerate}
\item[(i)] For any assignment $I:\vars(\Sinz(p_1,\dots,p_n))\to\{0,1\}$ such that $I$ satisfies $\Sinz(p_1,\dots,p_n)$,
any $1\leq j\leq n$ and $1\leq k\leq n$
we have 
$$
    I(\csf^k_j) = 1\quad \textrm{if, and only if,}\quad \quad\sum_{i=1}^j I(p_i) \geq k.
$$
\item[(ii)] For any $0/1$-sequence  $(a_1,\dots, a_n)\in\{0,1\}^n$ 
 there exists an assignment $I:\vars(\Sinz(p_1,\dots,p_n))\to \{0,1\}$ such that 
  $I(p_i) = a_i$, for $1\leq i \leq n$;
  $I$ satisfies $\Sinz(p_1,\dots,p_n)$; and 
  for any $r\leq n$ and $j\leq n$ if $\sum_{i=1}^j a_i \leq r$ then $I(\csf^k_j) = 0$, for $r<k\leq n$.
\end{enumerate}
\end{proposition}
It follows from Proposition~\ref{prop:sinz} that the formula
$(\Sinz(p_1,\dots,p_n)\land \lnot \csf^{r+1}_n)$ enforces the cardinality
constraint $p_1+\dots+p_n\leq r$. 

Formula (\ref{eq:sinz_card_1}) can be equivalently rewritten into clausal form:
\begin{eqnarray}
\label{eq:equiv_clause_1}%
  (\lnot \csf^k_j \lor \csf^k_{j-1}\lor \csf^{k-1}_{j-1}) & & \\
\label{eq:equiv_clause_2}%
  (\lnot \csf^k_j \lor \csf^k_{j-1}\lor p_j) & & \\
\label{eq:equiv_clause_3}%
  (\lnot \csf^k_{j-1}\lor \csf^k_j ) & & \\
\label{eq:equiv_clause_4}%
  (\lnot \csf^{k-1}_{j-1}\lor \lnot p_j\lor \csf^k_j), & & 
\end{eqnarray}
{where $1\leq k\leq n$, $1\leq j\leq n$}.

\smallskip

Notice that the original encoding in~\cite{Sinz05} only contains clauses
(\ref{eq:equiv_clause_3}) and (\ref{eq:equiv_clause_4}) due to a 
polarity-based optimisation based on Tseitin's~\shortcite{Tseitin70} renaming techniques.
One can see that
by unit propagation of clauses (\ref{eq:sinz_card_2}), (\ref{eq:sinz_card_3}) and
$(\lnot \csf^{r+1}_n)$ into (\ref{eq:equiv_clause_3}) and
(\ref{eq:equiv_clause_4}) we obtain exactly the set of clauses used
in~\cite{Sinz05}, which consists of $O(nr)$ clauses and requires $O(nr)$
auxiliary propositions. 
This optimisation leads to a relaxation in
item~(\textit{i}) of Proposition~\ref{prop:sinz}: Suppose that an assignment
$I$, which
satisfies 
(\ref{eq:sinz_card_2}),
(\ref{eq:sinz_card_3}),
(\ref{eq:equiv_clause_3}) and
(\ref{eq:equiv_clause_4}),
is such that $I(p_1)+\dots +I(p_n) \leq r' < r$.  Then the value of
$\csf^k_n$ under $I$ can still be true, as long as $r'<k\leq r$, while not
violating the cardinality constraint $p_1+\dots p_n\leq r$.  

For our purposes 
we require an unoptimised version of the encoding with 
a tighter restriction on the values of $\csf^k_j$ 
stated in Proposition~\ref{prop:sinz}.

\section{SAT Encoding of the discrepancy problem}\label{sec:encoding}

We say that a $\pm1$ sequence $x_1,\dots, x_n$ is \emph{$C$-bounded}, for some $C > 0$,
if  $|\sum_{i=1}^{j}x_i| \leq C$, for all $1\leq j \leq n$.
We extend the notion of $C$-boundedness to Boolean $0/1$ sequences using the 
relation between $\pm1$ sequences $x_1,\dots,x_n$ and Boolean
$\{0,1\}$-sequences $p_1,\dots, p_n$ defined as follows: $x_i = 2p_i-1$ (in other
words, $+1$ is encoded by the Boolean value \emph{true}, and $-1$ is encoded by
the Boolean value \emph{false}). Then a $0/1$ sequence 
$a_1,\dots, a_n$ is $C$-bounded if for every $j>0$ 
the disbalance between the number of $1$s in $a_1,\dots, a_j$ and the 
number of $0$s in $a_1,\dots, a_j$ is at most $C$.

We build our SAT encoding of the Erd\H{o}s discrepancy problem on 
the sequential counter-based encoding of cardinality constraints described
in Section~\ref{sec:sinz_card}. 
We illustrate our approach by the following consideration.
By Proposition~\ref{prop:sinz}, an arbitrary 
assignment of $0/1$ values to propositions $p_1,\dots,p_n$ can be uniquely 
extended to a model of $\Sinz(p_1,\dots,p_n)$. We 
denote this model as $I$.
Suppose that the value of some $\csf^k_j$ under $I$ is false.
Then, by Proposition~\ref{prop:sinz}, the sequence $I(p_1),\dots,I(p_j)$ 
contains at most $k-1$ occurrences of $1$ and so it contains at least $j-(k-1)$
occurrences of $0$. Therefore, the disbalance between the number of occurrences
of $0$ and the number of occurrences of $1$ is at least $j-2k+2$.  If
$j>2k-2+C$ then the sequence $I(p_1),\dots,I(p_j)$ is not $C$-bounded. 
Thus, in our encoding of  $C$-bounded sequences we need to exclude such a possibility.
This can be achieved by conjoining formula $\Sinz(p_1,\dots,p_n)$  with
\begin{equation}
\label{eq:equiv_clause_8}%
  (\csf^k_j), \quad \textrm{for $1\leq k \leq n$ and $2k-2+C < j \leq n$}.
\end{equation}

Similarly, if the value of $\csf^k_j$ is true, for some $j<2k-C$, then 
the number of $1$s exceeds the number of $0$s by more than $C$, so these
possibilities should also be excluded by
\begin{equation}
\label{eq:equiv_clause_6}%
  (\lnot \csf^k_j),\quad  \textrm{for $1\leq k\leq n$ and  $0\leq j<2k-C $}.
\end{equation}

To summarise, let propositional formula $\Eq^{C}(p_1,\dots,p_n)$ be the conjunction of $\Sinz(p_1,\dots,p_n)$,
with (\ref{eq:equiv_clause_8}) and
(\ref{eq:equiv_clause_6}).
Notice that, as in the case of  $\Sinz(p_1,\dots,p_n)$, 
we write (\ref{eq:equiv_clause_8}) and (\ref{eq:equiv_clause_6}) explicitly
for the ease of explanation. The proof of the following theorem can 
be found in Appendix~\ref{sec:proofs}.
\begin{theorem}\label{th:cbounded}
For any assignment $I:\{p_1,\dots,p_n\}\to\{0,1\}$ the following holds: 
there exists an extension of $I$ to $I': \vars(\Eq^{C}(p_1,\dots,p_n))\to\{0,1\}$ that is 
a model of $\Eq^{C}(p_1,\dots,p_n)$ if, and only if, 
the sequence $I(p_1),\dots, I(p_n)$ is $C$-bounded.
\end{theorem}

As (\ref{eq:sinz_card_1}) is logically equivalent to the set of clauses
(\ref{eq:equiv_clause_1})--(\ref{eq:equiv_clause_4}), formula 
$\Eq^{C}(p_1,\dots,p_n)$  is logically equivalent to 
the set of clauses $S$ consisting of 
(\ref{eq:equiv_clause_1})--(\ref{eq:equiv_clause_4}), 
(\ref{eq:sinz_card_2}), 
(\ref{eq:sinz_card_3}), 
(\ref{eq:equiv_clause_8})
and
(\ref{eq:equiv_clause_6}). 
Let $\CBound^C(p_1,\dots,p_n)$ be the result of applying unit propagation to $S$ in an 
exhaustive manner.
One can see that
the set $\CBound^C(p_1,\dots,p_n)$ contains less than $C\cdot n$ auxiliary propositions
and less than $4C\cdot n$ clauses.

\begin{example}
We construct $\CBound^2(p_1,\dots,p_5)$,
a clausal representation of the statement that the sequence
$p_1,\dots, p_5$ is $2$-bounded.
We also demonstrate how clauses (\ref{eq:sinz_card_2}), (\ref{eq:sinz_card_3}),  
(\ref{eq:equiv_clause_8})
and 
(\ref{eq:equiv_clause_6})
are unit propagated into clauses (\ref{eq:equiv_clause_1})--(\ref{eq:equiv_clause_4}).

Notice that for $k=1$ every instance of clause (\ref{eq:equiv_clause_1}) contains 
literal $\csf^0_{j-1}$, the only literal of the unit clause (\ref{eq:sinz_card_3}).
Thus, every instance of clause (\ref{eq:equiv_clause_1}) for $k=1$ is redundant.

For $k=2$ and $j=1$, clause (\ref{eq:equiv_clause_1}) contains $\lnot \csf^2_1$, the
only literal of the unit clause (\ref{eq:sinz_card_2}), so for 
$k=2$ and $j=1$, the instance of clause (\ref{eq:equiv_clause_1})  is also redundant.

For $k=2$ and $j=2$,  an instance of
the unit clause (\ref{eq:sinz_card_2}), namely $\lnot\csf^2_1$, unit propagates into 
(\ref{eq:equiv_clause_1}) resulting in a 2-CNF clause 
$(\lnot \ssf^2_2 \lor \ssf^1_1)$.

For $k=2$ and $j=3$, (\ref{eq:equiv_clause_1}) instantiates to 
$ (\lnot \ssf^2_3 \lor \ssf^2_2 \lor \ssf^1_2)$.

Finally, for $k=2$  and 
for $4\leq j\leq 5$, clause (\ref{eq:equiv_clause_1}) contains $\csf^1_{j-1}$, 
the only literal of the unit clause  (\ref{eq:equiv_clause_8}). Thus, for 
$k=2$ and $4\leq j\leq 5$ the instances of clause (\ref{eq:equiv_clause_1}) are redundant.

By a further consideration of cases, one can see that the set of all
non-redundant simplified instance of clause (\ref{eq:equiv_clause_1}), for
$1\leq k \leq 5$ and $j\leq j\leq 5$, consists of 
\begin{equation}\label{eq:ex1.1}
\begin{array}{l@{\quad}l}
(\lnot \ssf^2_2 \lor \ssf^1_1)                &  (\lnot \ssf^3_4 \lor \ssf^2_3)      \\[2pt]
(\lnot \ssf^2_3 \lor \ssf^2_2 \lor \ssf^1_2)  &  (\lnot \ssf^3_5 \lor \ssf^3_4 \lor \ssf^2_4).
\end{array}
\end{equation}
We group here the clauses in such a way that all clauses in one column correspond to the same
value of the parameter $k$.

Similarly, instances of clause (\ref{eq:equiv_clause_2}), for $1\leq k \leq 5$ and $1\leq j\leq 5$ 
are simplified with the help of unit clause
(\ref{eq:sinz_card_2}), (\ref{eq:sinz_card_3}), 
(\ref{eq:equiv_clause_8})  
and 
(\ref{eq:equiv_clause_6})
to
\begin{equation}\label{eq:ex1.2}
\begin{array}{l@{\qquad}l@{\qquad}l}
(\lnot \ssf^1_1 \lor p_1 )              &  (\lnot \ssf^2_2 \lor p_2)               & (\lnot \ssf^3_4 \lor p_4)               \\[2pt]
(\lnot \ssf^1_2 \lor \ssf^1_1 \lor p_2) &  (\lnot \ssf^2_3 \lor \ssf^2_2 \lor p_3) & (\lnot \ssf^3_5 \lor \ssf^3_4 \lor p_5). \\[2pt]
(\ssf^1_2 \lor p_3 )                    &  (\lnot \ssf^2_4 \lor \ssf^2_3 \lor p_4) &\\[2pt]
\mbox{}                                 &  (\ssf^2_3 \lor p_5)                     &\\[2pt]
\end{array}
\end{equation}
The set of all non-redundant simplified instances of clause (\ref{eq:equiv_clause_3}) consists of
\begin{equation}\label{eq:ex1.3}
\begin{array}{lll}
(\lnot \ssf^1_1 \lor \ssf^1_2) & (\lnot \ssf^2_2 \lor \ssf^2_3)     & (\lnot \ssf^3_4 \lor \ssf^3_5)\\[2pt]
                               & (\lnot \ssf^2_3 \lor \ssf^2_4)    &
\end{array}
\end{equation}
and set of all non-redundant simplified instances of clause (\ref{eq:equiv_clause_4}) consists of
\begin{equation}\label{eq:ex1.4}
\begin{array}{l@{\qquad}l@{\qquad}l@{\qquad}l}
(\lnot p_1 \lor \ssf^1_1) &  (\lnot \ssf^1_1 \lor \lnot p_2 \lor \ssf^2_2) &  (\lnot \ssf^2_2 \lor \lnot p_3)                                             & (\lnot \ssf^3_4 \lor \lnot p_5).\\[2pt]
(\lnot p_2 \lor \ssf^1_2) &  (\lnot \ssf^1_2 \lor \lnot p_3 \lor \ssf^2_3) &  (\lnot \ssf^2_3 \lor \lnot p_4 \lor \ssf^3_4)&\\[2pt]
\mbox{}                   &  (\lnot p_4 \lor \ssf^2_4)                     &  (\lnot \ssf^2_4 \lor \lnot p_5 \lor \ssf^3_5)& \\[2pt]
\end{array}
\end{equation}
Thus, the set $\CBound^2(p_1,\dots,p_5)$ consists of 26 clauses grouped in
(\ref{eq:ex1.1})--(\ref{eq:ex1.4}) above.
\end{example}

Any $\pm1$-sequence containing less than or equal to $C$ elements is always $C$-bounded.
It should be clear then that the discrepancy of a $\pm1$ sequence $x_1,\dots,
x_n$ is bounded by $C$ if, and only if, for every $d: 1\leq d \leq
\lfloor\frac{n}{C+1}\rfloor$ the subsequence $x_d, x_{2d}\dots, x_{\lfloor n/d
\rfloor\cdot d}$ is $C$-bounded.
Then we define
\begin{equation}\label{eq:edp}
\edp(C,n) = \bigwedge_{d=1}^{\lfloor\frac{n}{C+1}\rfloor}\CBound^C(p_d, p_{2d},\dots, p_{\lfloor n/d\rfloor\cdot d}).
\end{equation}
We assume here that for the different values of $d$ sets
$\CBound^C(p_d, x_{2d},\dots, p_{\lfloor n/d\rfloor\cdot d})$ share the same input
propositions $p_1,\dots, p_n$ but use different auxiliary propositions
$\csf^k_j$.  Then the following theorem is a direct consequence of
Theorem~\ref{th:cbounded}
\begin{theorem}\label{th:discrepancy}
For any assignment $I:\{p_1,\dots,p_n\}\to\{0,1\}$ the following holds: 
there exists an extension of $I$ to $I': \vars(\edp(C,n))\to\{0,1\}$ that is 
a model of $\edp(C,n)$ if, and only if, 
$I(p_1),\dots, I(p_n)$ encodes a 
$\pm1$ sequence $x_1,\dots, x_n$ of length $n$ and discrepancy at most $C$.
\end{theorem}

We conclude this section with a description of two optimisations, which we
present in the form of propositions. Both  reduce significantly the size of 
the unsatisfiability certificate and have some noticeable effect on the running
time.  The first optimisation allows one to remove the `don't care'
propositions, which do not affect the satisfiability of the problem. The second
optimisation breaks the symmetry in the problem.

\begin{proposition}\label{prop:oddeven}
Suppose that a sequence $a_1,\dots, a_n$ is $C$-bounded and either
$n$ is odd and $C$ is even or 
$n$ is even and $C$ is odd.
Then for an arbitrary value $b$ the sequence
$a_1,\dots, a_n, b$ is $C$-bounded.
\end{proposition}
\begin{proof}
It suffices to notice that $|\sum_{i=1}^j a_i|$ is odd if, and only if, $j$  is
odd. Thus, under the conditions of the proposition,  
$|\sum_{i=1}^j a_i|\leq C-1$, and the sequence can be extended arbitrarily.  
\end{proof}
\begin{proposition}\label{prop:symmetry}
There exists a $\pm1$ sequence $x_1,\dots, x_n$ of length $n$ and discrepancy
at most $C$ if, and only if, there exists a $\pm1$ sequence $y_1,\dots, y_n$ of
length $n$ and discrepancy at most $C$, in which $y_l = 1$, for some arbitrary
but fixed value of $l$.
\end{proposition}
\begin{proof}
The problem is symmetric, that is, the discrepancy of  $x_1,\dots, x_n$ 
is bounded by $C$ if, and only if, the discrepancy of $-x_1,\dots, -x_n$ is bounded by $C$.
\end{proof}

\subsection{SAT Encoding of  Multiplicativity}
Multiplicativity and complete multiplicativity of $\pm1$ sequences can be
encoded in SAT in a rather straightforward way. Assuming that a Boolean
sequence $p_1,\dots p_n$ encodes a $\pm1$ sequence $x_1,\dots,x_n$ so that the
logical value $1$  encodes the numerical value $+1$ and the logical value $0$
encodes the numerical value $-1$, a SAT encoding of the fact that 
$x_{j\cdot k} = x_j\cdot x_k$ 
is captured by the following clauses, which 
enumerate all four combinations of values of $x_j$ and $x_k$:
\begin{equation}\label{eq:mult2}
\textsf{prod}_{j,k} = \{
(\lnot p_{j}\lor \lnot p_{k}\lor p_{j\cdot k}),
(p_{j}\lor p_{k}\lor p_{j\cdot k}),
(\lnot p_{j}\lor  p_{k}\lor \lnot p_{j\cdot k}),
(p_{j}\lor \lnot p_{k}\lor \lnot p_{j\cdot k})\} 
\end{equation}
Then multiplicativity of $x_1,\dots x_n$ is captured by instances of (\ref{eq:mult2}) for
all coprime pairs $i$ and $j$; and, by (\ref{eq:multdef2}), complete multiplicativity
of the sequence $x_1,\dots, x_n$ is captured by instances of (\ref{eq:mult2}) for $j$
and $k$ such that every product $j\cdot k$ is generated only once.

For complete multiplicativity further optimisation is possible due to the fact
that in any such sequence $x_{j^2} = +1$ for any $j\in\natP$.  
It can be seen that the CNF formula $\cmult_i$  defined below 
expresses the complete multiplicativity condition on  $i$,
for every $i: 1\leq i \leq n$.
$$
\cmult_i = 
\left\{
\begin{array}{ll}
 \emptyset  & \textrm{ if $i$ is prime}\\[5pt]
  \{(p_i)\}  & \textrm{ if $i = j^2$, for some $j\geq 1$} \\[5pt]
%
  \mathsf{prod}_{j,k}\quad &\textrm{ if none of the cases above applies}\\
    & \textrm{ and $j,k$ are some non-trivial divisors of $i$.}
\end{array}
\right.
$$
Then we define two sets of clauses 

$$
\Mult(C,n) = \edp(C,n)\cup\mathop{\bigcup}_{{\substack{1< j\leq n, 1< k \leq n\\ j,k\textrm{ are coprime}\\j\cdot k \leq n}}} \textsf{prod}_{j,k}
$$
and
$$
\CMult(C,n) = \edp(C,n)\cup\bigcup_{i=1}^n \cmult_i.
$$
The following statement is a direct consequence of Theorem~\ref{th:discrepancy}.
\begin{theorem}\label{th:mult_discrepancy}
For any assignment $I:\{p_1,\dots,p_n\}\to\{0,1\}$ the following holds: 
there exists 
an extension of $I$ to $I': \vars(\Mult(C,n))\to\{0,1\}$
(or an extension of $I$ to $I': \vars(\CMult(C,n))\to\{0,1\}$),
which is 
a model of $\Mult(C,n)$ (or $\CMult(C,n)$, respectively) if, and only if, 
$I(p_1),\dots, I(p_n)$ encodes a multiplicative (or completely multiplicative, respectively)
$\pm1$ sequence $x_1,\dots, x_n$ of length $n$ and discrepancy at most $C$.
\end{theorem}
Finally we notice that the completely multiplicative case 
can be optimised based on the following observation.
\begin{proposition}\label{prop:cmult_opt}
The discrepancy of a completely multiplicative $\pm1$ sequence $x_1,\dots, x_n$
is bounded by $C$, for some $C>0$, if, and only if,  $x_1,\dots, x_n$ is
$C$-bounded.
\end{proposition}
\begin{proof}
The necessary condition is trivial by definition of discrepancy. For the sufficient
condition we show that for any $C$-bounded sequence $x_1,\dots,x_n$ and any $d > 1$ 
the subsequence $x_{d}, x_{2d},\dots, x_{\lfloor n/d\rfloor \cdot d}$ is $C$-bounded.
Let $1\leq j \leq {\lfloor n/d\rfloor}$. Then
$
|\sum_{i=1}^j x_{i\cdot d}| = 
|\sum_{i=1}^j (x_{i}\cdot x_{d}) | = 
|x_d\cdot\sum_{i=1}^j x_{i} | = 
|\sum_{i=1}^j x_{i} | \leq C$.
\end{proof}
\section{Results}\label{sec:experiments}
In our experiments we use \Treengeling, a parallel cube-and-conquer flavour of the
\Lingeling SAT solver~\cite{lingeling}
version {aqw}, the winner of the \emph{SAT-UNSAT} category of the SAT'13
competition~\cite{sat13}, and the \Glucose solver~\cite{Glucose} version
3.0, the winner of the \emph{certified UNSAT} category of the SAT'13
competition~\cite{sat13}.  All experiments were conducted on PCs equipped with
an Intel Core i5-2500K CPU running at 3.30GHz and 16GB of RAM.

In our first series of experiments we investigate the discrepancy of
unrestricted $\pm1$ sequences.  We encode\footnote{The problem generator and results can be found at~\url{http://www.csc.liv.ac.uk/~konev/edp/}} the existence of a $\pm1$ discrepancy
$C$ sequence of length $n$ into SAT as described in
Section~\ref{sec:encoding}. We deploy both optimisations described in
Proposition~\ref{prop:oddeven} and Proposition~\ref{prop:symmetry}. 
We choose as $l$, for which we fix $x_l$ to be $+1$, colossally abundant numbers~\cite{abundant},
which have many divisors and thus contribute to many homogeneous sequences.
Specifically for $C=2$, the choice of
$l=120$ is more beneficial for satisfiable instances; however,
$l={60}$ results in a better reduction of the size of the unsatisfiability proof described below.
For consistency of presentation, we use $l=60$ in all our experiments for
$C=2$.

For $C=2$ we establish that the maximal length of a $\pm 1$
sequence of discrepancy $2$ is $1160$. 
The CNF formula $\edp(2, 1160)$ contains $11824$ propositions and  $41\,884$ clauses.
It takes the \Treengeling system about 
${430}$~seconds to find an example of such a sequence 
on our hardware configuration.  
One of the  sequences of length $1160$ of discrepancy
$2$ can be found in Appendix~\ref{sec:sequence}.

When  applied to the CNF formula $\edp(2, 1161)$, which contains $11847$
propositions and  $41\,970$ clauses, \Treengeling reports unsatisfiability. 
In order to corroborate this statement, we also use \Glucose.  It takes the
solver about ${800}$~seconds to  generate a DRUP certificate of
unsatisfiability.  The correctness of the generated unsatisfiability
certificate has been independently verified with the \druptrim
tool~\cite{druptrim}. 

The size of the certificate is about $\textrm{1.67}$~GB,
and the time needed to verify the certificate is comparable with the time
needed to generate it.  The RUP unsatisfiability certificate, that is the DRUP
certificate with all information on the deleted clauses stripped, is $850.2$MB;
it takes \druptrim about five and a half hours to verify it.  Combined with
Theorem~\ref{th:discrepancy},  these two experiments yield  a computer proof of
the following statement.
\begin{theorem}\label{th:edp2}
The length of a maximal  $\pm1$ sequence of discrepancy $2$ is $1160$.
\end{theorem}
Thus we prove that the Erd\H{o}s discrepancy conjecture holds true for $C=2$.

\smallskip

When applied to $\edp(3,n)$ for increasing values of $n$ our method could only
produce sequences of discrepancy $3$ of length in the region of $14\,000$, even
though solvers were allowed to run for weeks.  Since both  multiplicativity
and complete multiplicativity restrictions reduce severely the search space, in
hope for better performance, we perform the second series of experiments to
investigate the discrepancy bound for multiplicative and
completely multiplicative sequences. Notice that the optimisation
described in Proposition~\ref{prop:symmetry} is not applicable in this case as the
fact that $x_1,\dots, x_n$ is multiplicative does not imply that $-x_1,\dots,
-x_n$ is.

We saw in Example~\ref{ex:edp1} that multiplicative sequences of discrepancy
$1$ are longer than completely multiplicative sequences. The longest completely
multiplicative sequence of discrepancy $2$ is known to contain $246$
elements~\cite{Polymath2}; tests with $\Mult$ show that 
the longest multiplicative sequence of discrepancy $2$ has $344$ elements.
Thus it wouldn't be unreasonable to assume that 
the longest multiplicative discrepancy $3$ sequence is longer
than the longest completely multiplicative one, but is probably harder to find.
It turns out that this expectation is wrong on both accounts.

We establish that the length of a maximal $\pm 1$ completely multiplicative
discrepancy $3$ sequence coincides with the length of a maximal $\pm1$
multiplicative discrepancy $3$ sequence and is equal to $127\,645$. 
It takes \Treengeling about one~hour and fifty minutes to find
a satisfying assignment to $\CMult(3, 127\, 645)$, 
which contains $3\,484\,084$ propositions and $13\,759\,785$ clauses,
and about one hour and thirty five minutes to find a satisfying 
assignment to $\Mult(3, 127\, 645)$, 
which also contains $3\,484\,084$ propositions but $14\,813\,052$ clauses.

It takes the \Glucose solver just under eight hours to  
generate an approximately 1.95~GB DRUP  proof of 
unsatisfiability for $\CMult(3, 127\,646)$, which contains $3\,484\,084$ propositions
and $13\,759\,809$ clauses,
and about nine and a half hours to generate an approximately
3.78~GB DRUP proof of unsatisfiability for 
$\Mult(3, 127\,646)$,
which again contains the same number of propositions but $14\,813\,076$ clauses.
%

The optimisation of Proposition~\ref{prop:cmult_opt} leads
reduction in the problem size for the completely multiplicative case 
(446\,753~propositions and 1\,738\,125~clauses for length
127\,645 and 446\,759~propositions and 1\,738\,149~clauses for
length 127\,646) and a significant reduction in the 
\Treengeling running time (about 20 and 30 minutes, respectively); however, it 
does not 
reduce the size of the DRUP certificate for the 
unsatisfiable case, which is about 2.22~GB. This lack of reduction in
size is due to the fact that unit propagation steps are not recorded 
as part of a DRUP certificate.

So we get a computer-aided proof of another sharp bound on the sizes of maximal
sequences of bounded discrepancy.
\begin{theorem}\label{th:edp3mult}
The length of a maximal multiplicative 
$\pm1$ sequence of discrepancy $3$  equals to 
the length of a maximal completely multiplicative 
$\pm1$ sequence of discrepancy $3$  and is 
$127\,645$.
\end{theorem}
Unrestricted sequences of discrepancy $3$ can still be longer than $127\,646$: by requiring that
only first $127\,600$ elements of a sequence are completely multiplicative, we
generate a 130,000 long EDP3 sequence in about one hour and fifty minutes
thus establishing a slightly better lower bound on the length of 
$\pm1$ sequences of discrepancy $3$ than the one from Theorem~\ref{th:edp3mult}. 
The solvers struggle to expand it much further. Notice that the optimisation of Proposition~\ref{prop:cmult_opt} is not applicable here. 

We summarise known facts about discrepancy of unrestricted, multiplicative and
completely multiplicative sequences in Table~\ref{tab:results}. We highlight in boldface 
cases where the lengths of maximal sequences of different kinds are equal.
\begin{table}[t]\centering
{
\newcommand{\ra}[1]{\renewcommand{\arraystretch}{#1}}
\begin{tabular}{@{}lllll@{}}
\toprule
Discrepancy & & Completely          & Multiplicative    & Unconstrained\\
bound       & & multiplicative      &                   &              \\
\midrule                                                        
C = 1       & & 9                   & \textbf{11}       & \textbf{11}           \\
C = 2       & & 246                 & 344               & 1160         \\ 
C = 3       & & \textbf{127\,645}   & \textbf{127\,645} & $>$130\,000\\ 
\bottomrule
\end{tabular}
}
\caption{Maximal length of $\pm1$ sequences of bounded discrepancy\label{tab:results}}
\end{table}

\section{Conclusions}\label{sec:conclusion}
We have demonstrated that SAT-based methods can be used to tackle the
longstanding mathematical questions related to discrepancy of $\pm 1$ sequences. 
Not only were we able to identify the exact boundary between satisfiability and
unsatisfiability of the encoding of the EDP for $C=2$, thus identifying the
longest sequences of discrepancy $2$, but also we have established the
surprising fact that the lengths of the longest multiplicative and completely
multiplicative sequences of discrepancy $3$ coincide. 
The latter result helps to establish a novel lower bound on the length of the longest
discrepancy $3$ sequence. 

There is, however, a noticeable asymmetry in our findings.  The fact that a
sequence of length 1160 has discrepancy $2$ can be relatively easily checked
manually.  It is harder but not impossible to verify the correctness of the
discrepancy bound for $127\, 645$-long sequences.
On the other hand, even though improvements to our method shortened the
Wikipedia-size 13~GB proof reported in~\cite{KLArx14} 
more than 
tenfold passing
thus the psychological barrier of $1$~GB, it is still probably one of the
longest proofs of a non-trivial  mathematical result, and it is equally
improbable that a mathematician would verify by hand ten billions or half a
billion of automatically generated proof lines. 
It should be noted that this gigantic proof is a formal 
proof in a well-specified proof system and, as such, it has been verified by 
a third-party tool, and it can potentially be translated into a proof assistant such as Coq~\cite{coq}.
The reduction of proof size
will be useful for any future analysis of the proof in an attempt to
identify patterns and lemmas and produce a compact proof more amenable for
human comprehension. 

\paragraph{Acknowledgements}
The authors would like to thank
Armin Biere,
Marijn Heule,
Pascal Fontaine,
Laurent Simon and
Laurent Th\'ery
and 
Donald Knuth
for helpful  discussions, comments and ideas following the publication of the
preliminary version of this paper~\cite{KLArx14}.

\appendix

\section{Proofs of technical results}\label{sec:proofs}
In this section we give proofs of the technical results used in the main text.
We re-state propositions and theorems here for the reader's convenience.

\smallskip 

\textsc{Proposition~~\ref{prop:sinz}. }
{\it
Let $\Sinz(p_1,\dots,p_n)$ be as defined above.  Then 
\begin{enumerate}
\item[(i)] For any assignment $I:\vars(\Sinz(p_1,\dots,p_n))\to\{0,1\}$ such that $I$ satisfies $ \Sinz(p_1,\dots,p_n)$,
any $1\leq j\leq n$ and $1\leq k\leq n$
we have 
$$
    I(\csf^k_j) = 1\quad \textrm{if, and only if,}\quad \quad\sum_{i=1}^j I(p_i) \geq k.
$$
\item[(ii)] For any $0/1$-sequence  $(a_1,\dots, a_n)\in\{0,1\}^n$ 
 there exists an assignment $I:\vars(\Sinz(p_1,\dots,p_n))\to \{0,1\}$ such that 
  $I(p_i) = a_i$, for $1\leq i \leq n$;
  $I$ satisfies $\Sinz(p_1,\dots,p_n)$; and 
  for any $r\leq n$ and $j\leq n$ if $\sum_{i=1}^j a_i \leq r$ then $I(\csf^k_j) = 0$, for $r<k\leq n$.
\end{enumerate}
}
\begin{proof}
\begin{enumerate}
\item[(i)] The proof proceeds by induction on the lexicographical partial 
order $\prec$ on pairs of non-negative integers:
$(j,k) \prec (j',k')$  iff $j < j' \lor ((j = j') \land (k < k'))$. 
Fix some $n \ge 1 $.  

Consider cases: 

\begin{itemize}
\item   Suppose that  { $j=k=1$}. 
Then formula~(\ref{eq:sinz_card_1}), one
of the conjuncts of $\Phi^{n}(p_{1}, \ldots, p_{n})$, instantiates to
$s_{1}^{1} \leftrightarrow (s_{0}^{1} \lor (s_{0}^{0} \land p_{1})$. Therefore,
for an assignment $I$ such that 
$I $ satisfies $ \Phi^{n}(p_{1}, \ldots, p_{n})$
we have 
$I(s_{1}^{1} \leftrightarrow (s_{0}^{1} \lor (s_{0}^{0} \land p_{1})) = 1$.  
Furthermore, for the satisfying assignment we have
$I(s_{0}^{1}) = 0$ 
and
$I(s_{0}^{0}) = 1$. 
It follows then 
that
$I(s_{1}^{1}) = I(p_{1})$, 
which is equivalent to the statement of the proposition for
the case $k=j=1$.

\item Suppose that {\ $j=1$ and $k > 1$.} 
For a satisfying assignment $I$ we have $I(s_{1}^{k}) = 0$ 
(as (\ref{eq:sinz_card_2}) is a conjunct of $\Phi^{n}(p_{1}, \ldots, p_{n})$). 
On the other hand for $k > 1$ we have $\sum_{i=1}^1 I(p_i) <  k$.
Thus the statement of the proposition holds true in this case. 

\item  Suppose that 
{$j > 1$, $k \ge 1$.} For a satisfying assignment $I$ we
have $I(s_{j}^{k}) = 1$ if and only if  $I(s_{j-1}^{k}) =1$  or
$I(s_{j-1}^{k-1} \land p_{j}) = 1$ (by satisfaction of (\ref{eq:sinz_card_1})). By induction
hypothesis the later is equivalent to $\sum_{i=1}^{j-1} I(p_i) \geq k$ or
$\sum_{i=1}^{j-1} I(p_i) \geq k-1$ and $I(p_{j})=1$, which in turn is
equivalent to   $\sum_{i=1}^{j} I(p_i) \geq k$. 
\end{itemize}

\item[(ii)] First notice that any assignment $I_{p}: \{p_{1}, \ldots, p_{n})
\rightarrow \{0,1\}$ can be extended  in a unique way to the assignment
$I:\vars(\Sinz(p_1,\dots,p_n))\to \{0,1\}$. 
Indeed, satisfaction of (\ref{eq:sinz_card_2}) and (\ref{eq:sinz_card_3}) defines
uniquely the values of satisfying assignment $I$ on $s_{j}^{k}$ for the cases
$0\le j < k \le n$ and $k=0; 0 \le j \le n$, respectively.  Further, using
satisfaction condition for (\ref{eq:sinz_card_1}) the values of $I$ on the remaining  variables
$s_{j}^{k}$ with $1 \le k \le n, 1 \le j \le n$  are defined uniquely by
induction on $\prec$.  The remaining condition, that is for any $r\leq n$ and $j\leq n$ if
$\sum_{i=1}^j a_i \leq r$ then $I(\csf^k_j) = 0$, for $r<k\leq n$, now follows
from item~(\textit{i}) above.
\end{enumerate}
\end{proof}


\smallskip 

\textsc{Theorem~~\ref{th:cbounded}. }
{\it
For any assignment $I:\{p_1,\dots,p_n\}\to\{0,1\}$ the following holds: 
there exists an extension of $I$ to $I': \vars(\Eq^{C}(p_1,\dots,p_n))\to\{0,1\}$ that is 
a model of $\Eq^{C}(p_1,\dots,p_n)$ if, and only if, 
the sequence $I(p_1),\dots, I(p_n)$ is $C$-bounded.
}
\begin{proof}
$\Longleftarrow$
Assume that for an assignment $I: \{p_{1}, \ldots, p_{n}\} \rightarrow \{0,1\}$ 
the sequence $I(p_{1}), \ldots, I(p_{n})$ is $C$-bounded. 
By item (\textit{ii}) of Proposition~\ref{prop:sinz}, $I$
can be extended to an assignment $I': vars(\Phi(p_{1}, \ldots, p_{n}) \rightarrow \{0,1\}$ 
such that 
$I'(p_{i}) = I(p_{i})$, for all $1 \le i \le n$; 
$I'$ satisfies $\Phi(p_{1}, \ldots , p_{n})$; 
and, furthermore, $I'(s_{j}^{k}) = 1$  if, and only if, $\sum_{i=1}^{j} I'(p_{i})  \ge k$. 
We show that $I'$ is, in fact,  a model of $\Psi^{C}(p_{1}, \ldots, p_{n})$. It suffices 
to demonstrate that clauses
(\ref{eq:equiv_clause_8}) and 
(\ref{eq:equiv_clause_6}) are true under $I'$.
\begin{enumerate}
\item
Assume to the contrary that $I'(s_{j}^{k}) = 0$,  for some $1 \le k \le n$ and $2k-2+C < j \le n$. 
Then by Proposition~\ref{prop:sinz} we have 
\begin{equation}\label{eq:proof2}
    \sum_{i=1}^{j}I'(p_{i}) < k.
\end{equation}
As the number of occurrences of $1$ in  $I'(p_{1}), \ldots,I'( p_{n})$ is 
$\sum_{i=1}^{j} I'(p_{i})$
(and, hence,
the number of occurrences of $0$ in  $I'(p_{1}), \ldots,I'( p_{n})$ is 
$\sum_{i=1}^{j} (1-I'(p_{i}))$), 
the difference  between the number of occurrences of $0$ and the number of occurrences  of $1$ is
at least $\sum_{i=1}^{j} (1-I'(p_{i})) - \sum_{i=1}^{j} I'(p_{i}) = j - 2\sum_{i=1}^{j} I'(p_{i})$.

As we supposed that $j > 2k-2+C$, the difference  between the number of
occurrences of $0$ and the number of occurrences  of $1$ exceeds $ 2k -2 + C -
2\sum_{i=1}^{j} I'(p_{i}) \geq  2k -2  + C - 2(k-1) = C$ (by  (\ref{eq:proof2}))
and therefore  the sequence
$I'(p_{1}), \ldots, I'(p_{n})$ and hence $I(p_{1}), \ldots, I(p_{n})$ is not
$C$-bounded contradicting our assumption.

\item  Assume to the contrary that $I'(s_{j}^{k}) = 1$  for some $1 \le k \le n$ and $0\le j < 2k-C$.  
Then $\sum_{i=1}^{j}I'(p_{i}) \ge  k$  (by Proposition ~\ref{prop:sinz})
Now we estimate the difference between the number of occurrences 
of $1$ in $I'(p_{1}) \ldots I'(p_{n})$ and the number
of occurrences of $0$ in $I'(p_{1}) \ldots I'(p_{n})$. 
We have 

$$
\sum_{i=1}^{j}I'(p_{i}) - \sum_{i=1}^{j}(1-I'(p_{i})) = 2( \sum_{i=1}^{j}I'(p_{i})) - j \ge 2k - j > 2k- (2k-C) = C.
$$ 

It follows  that  the sequence $I(p_{1}), \ldots, I(p_{n})$ is not $C$-bounded
contradicting our assumption. 

\end{enumerate}


$\Longrightarrow$ Consider an assignment 
$I : \{p_{1}, \ldots, p_{n}\} \rightarrow \{0,1\}$ and assume that its
extension to $I': vars(\Psi^{C}(p_{1}, \ldots, p_{n})) \rightarrow \{0,1\}$
is a model of  $\Psi^{C}(p_{1}, \ldots, p_{n})$. Now we are to show that
$I(p_{i}), \ldots, I(p_{n})$ is $C$-bounded.  Assume to the contrary that it is
not $C$-bounded. As $I'$ is an extension of $I$ we have that 
$I'(p_{i}), \ldots, I'(p_{n})$ is also not $C$-bounded. 
Then we have $\sum_{i=1}^{j} I'(p_{j}) > C$  for some $j$. Consider two cases: 

\begin{enumerate}
\item The number of occurrences of $1$  in $I'(p_{i})$, for $1 \le i \le j$,
exceeds the number of occurrences of $0$ by more than $C$.  
Let $\sum_{i=1}^{j} I'(p_{i}) = k$. Then we have 
$$
\sum_{i=1}^{j} I'(p_{i}) - \sum_{i=1}^{j}(1 - I'(p_{i})) > C,
$$ that is $ k-j+k > C$ and $j < 2k-C$.  Now we have $I'(s_{j}^{k})=0$, by the satisfaction condition for (\ref{eq:equiv_clause_8}),
and $I'(s_{j}^{k}) = 1$, by  $\sum_{i=1}^{j} I'(p_{i}) = k$. A contradiction.

\item The number of occurrences of $0$ in $I'(p_{i})$, for  $1 \le i \le j$,
exceeds the number of occurrences of $1$ by more than $C$. 
Let $\sum_{i=1}^{j} I'(p_{i}) = k-1$.  Then we have 
$$
    \sum_{i=1}^{j}(1- I'(p_{i})) - \sum_{i=1}^{j}I'(p_{i}) > C,
$$ 
that is, $j - 2(k-1) > C$ and $j-2k+2 > C$.  
Then, by the satisfaction condition for  (\ref{eq:equiv_clause_6}), we have 
$I'(s_{j}^{k}) = 1$   and, $\sum_{i=1}^{j} I'(p_{i}) = k-1$, we have  $I'(s_{j}^{k}) = 0$.  A contradiction. 
\end{enumerate}
\mbox{}
\end{proof}

\section{A sequence of length 1160 and discrepancy 2}\label{sec:sequence}
We give a graphical representation of one of the sequences of length $1160$ 
obtained from the satisfying assignment computed with the \Treengeling solver.
Here $+$ stands for $+1$ and $-$ for $-1$, respectively.

\smallskip 

\noindent{}{\tt 
- + + - + - - + + - + + - + - - + - - + + - + - - + - - +\\
+ - + - - + + - + + - + - + + - - + + - + - - - + - + + - \\
+ - - + - - + + + + - - + - - + + - + - - + + - + + - - - \\
- + + - + + - + - + + - - + + - + - + - - - + + - + - - +\\
+ - + + - + - - + + - + - - + - - - + - + + - + - - + + - \\
+ + - + - - + - - + + - + + - + - - + + - + - - + + + - +\\
- + - - - - + + + - + - - + - - + + + - - - + + - + + - + \\
- - + - - + + + - - + - + - + - - + - + + + - + + - + - - \\
+ - - + + - + - - + + - + + - + - - + - - + + - - + + + -\\
- - + + + - + - - - + + - + - - + + - - + - + - - + - + + \\
+ - + - - + + - + + - + - - + + - + - - + - - + + - + - -\\
+ + - - + - + + - + - + - - + - + - + + - + - - + + - + -\\
- + - - + + - + - + - + + - + - + - + + - - - + - + - - + \\
+ + + - - + - - - + + - + - + + - + - - + + - + - - + - -\\
+ + - + - - + + + + - - + - - - + - + + + + - - + - - + +\\
- + + - + - - + + - + - - + - - + + - + - - + + - + + - +\\
- - + + - + - - + - - + + - + + - + - - - - + + + - + - - \\
+ + - - + + + - - - + - + + - + - - + - + + - - - + - + +\\
- + + - + - - + - - + + - - + + + + - + - - + - - + - - + \\
+ + + - - + - - + + + - - - + + - + + - + - - + + - - + - \\
+ - - + - - + + - + + - + - - + - - + - + + + - + + - + - \\
- + - - + + - - + - + + - + + - + - - + - - + - - + + - + \\
+ - + - - + + + - - - + + - + - - + + - + + - - - + + + - \\
- - + + - + + - - - - + + + - - + - + + - + - - + - - + +\\
- + - - + + - + + - + - + + - - + + - - + + - - - - + + + \\
- + + - - + + - - - - + + - + + + - - + + - - - + + + - -\\
- - + - + - + + - + + - + + - + - + - - - - + + + - - + +\\
- + - - + + - + + - + - - + - - + - - + + - + - - + + - + \\
+ - + - - + + - - + - + - - + - + - + - + + + + - - - + -\\
+ - + + - - + - - + - + - + - + + - + - + + + - - + - + - \\
- + - - + - + + + - - + - + + + - - - + + - + - - + - - +\\
+ - + + - - + + - - - + + - + - + + - - + + - + - - - + - \\
+ + - + - - + - + + - - + + - + - - + + - + - - + - + + + \\
- + - - + + - - + - + - + + + - - + - + - - + + - + + - +\\
- - + - - + - + + - - - + - + + - + - + + - - + + - + - -\\
+ + + - + - - - - + + - - + - + + - + - + + - - + + - + - \\
- + + - + - + + - - + + - + - - - + - + + - + - - + + + - \\
- - - + - + - + + - - + + - + - - + + - + + - + + - + - - \\
+ - - + - - + + + + - - - + + - - - + - + - + + - + - + + \\
+ - - + - + + - - + - + - - + - + - + + - - - + + + - + + \\
}


\begin{thebibliography}{}

\bibitem[\protect\citeauthoryear{Alaoglu and Erd{\H{o}}s}{1944}]{abundant}
Leonidas Alaoglu and Paul Erd{\H{o}}s.
\newblock On highly composite and similar numbers.
\newblock {\em Transactions of the American Mathematical Society},
  56(3):448--469, 1944.

\bibitem[\protect\citeauthoryear{Alon}{1992}]{DBLP:journals/dam/Alon92}
Noga Alon.
\newblock Transmitting in the {\it n}-dimensional cube.
\newblock {\em Discrete Applied Mathematics}, 37/38:9--11, 1992.

\bibitem[\protect\citeauthoryear{Apostol}{1976}]{apostol}
Tom~M. Apostol.
\newblock {\em Introduction to analytic number theory}.
\newblock Undergraduate Texts in Mathematics. Springer, New York, NY, USA,
  1976.

\bibitem[\protect\citeauthoryear{Armand \bgroup \em et al.\egroup }{2011}]{coq}
Micha{\"e}l Armand, Germain Faure, Benjamin Gr{\'e}goire, Chantal Keller,
  Laurent Th{\'e}ry, and Benjamin Werner.
\newblock A modular integration of SAT/SMT solvers to coq through proof
  witnesses.
\newblock In Jean-Pierre Jouannaud and Zhong Shao, editors, {\em Certified
  Programs and Proofs - First International Conference, CPP 2011, Kenting,
  Taiwan, December 7-9, 2011. Proceedings}, volume 7086 of {\em Lecture Notes
  in Computer Science}, pages 135--150. Springer, 2011.

\bibitem[\protect\citeauthoryear{Audemard and Simon}{2013}]{Glucose}
Gilles Audemard and Laurent Simon.
\newblock Glucose 2.3 in the {SAT} 2013 {C}ompetition.
\newblock In {\em Proceedings of SAT Competition 2013}, pages 42--43, Helsinki,
  2013. University of Helsinki.

\bibitem[\protect\citeauthoryear{Balint \bgroup \em et al.\egroup
  }{2013}]{sat13}
Adrian Balint, Anton Belov, Marijn J.~H. Heule, and Matti J\"arvisalo, editors.
\newblock {\em Proceedings of {SAT} competition 2013}. University of Helsinki,
  2013.

\bibitem[\protect\citeauthoryear{Beck and Chen}{1987}]{Beck}
J{\'{o}}zsef Beck and William W.~L. Chen.
\newblock {\em Irregularities of Distribution}.
\newblock Cambridge University Press, Cambridge, 1987.

\bibitem[\protect\citeauthoryear{Beck and S{\'o}s}{1995}]{BeckSos}
J{\'o}sef Beck and Vera~T. S{\'o}s.
\newblock Discrepancy theory.
\newblock In Ronald.~L. Graham, Martin Gr{\"o}tschel, and L{\'a}sl{\'o}
  Lov{\'a}sz, editors, {\em Handbook of combinatorics}, volume~2, pages
  1405--1446. Elsivier, Amsterdam, 1995.

\bibitem[\protect\citeauthoryear{Biere \bgroup \em et al.\egroup
  }{2009}]{satHB}
Armin Biere, Marijn Heule, Hans van Maaren, and Toby Walsh, editors.
\newblock {\em Handbook of Satisfiability}, volume 185 of {\em Fronteers in
  Artificial Intelligence and Applications}.
\newblock IOS Press, Amsterdam, 2009.

\bibitem[\protect\citeauthoryear{Biere}{2013}]{lingeling}
Armin Biere.
\newblock {L}ingeling, {P}lingeling and {T}reengeling entering the {SAT}
  {C}ompetition 2013.
\newblock In {\em Proceedings of SAT Competition 2013}, pages 51--52, Helsinki,
  2013. University of Helsinki.

\bibitem[\protect\citeauthoryear{Borwein \bgroup \em et al.\egroup
  }{2010}]{BCC10}
Peter Borwein, Stephen K.~K. Choi, and Michael Coons.
\newblock Completely multiplicative functions taking values in {$\{ 1,-1 \}$}.
\newblock {\em Transactions of the American Mathematical Society},
  362(12):6279--6291, 2010.

\bibitem[\protect\citeauthoryear{Chazelle}{2000}]{Chazelle}
Bernard Chazelle.
\newblock {\em The Discrepancy Method: Randomness and Complexity}.
\newblock Cambridge University Press, New York, 2000.

\bibitem[\protect\citeauthoryear{{Erd\H{o}s}}{1957}]{Unsolved}
Paul {Erd\H{o}s}.
\newblock Some unsolved problems.
\newblock {\em The Michigan Mathematical Journal}, 4(3):291--300, 1957.

\bibitem[\protect\citeauthoryear{Goldberg and Novikov}{2003}]{GoldbergN03}
Evguenii~I. Goldberg and Yakov Novikov.
\newblock Verification of proofs of unsatisfiability for {CNF} formulas.
\newblock In {\em Proceedings of Design, Automation and Test in Europe
  Conference and Exposition (DATE 2003), 3-7 March 2003, Munich, Germany},
  pages 10886--10891, 2003.

\bibitem[\protect\citeauthoryear{Gowers}{2009}]{gowersblog}
Timothy Gowers.
\newblock Is massively collaborative mathematics possible?, 2009.
\newblock 
Retrieved April, 10 2014 from
  \url{http://gowers.wordpress.com/2009/01/27/is-massively-collaborative-mathematics-possible/}

\bibitem[\protect\citeauthoryear{Gowers}{2013}]{gowers}
Timothy Gowers.
\newblock Erd{\H{o}}s and arithmetic progressoins.
\newblock In {\em Erd{\H{o}}s Centennial conference}, 2013.
\newblock
Retrieved April, 10 2014 from
  \url{http://www.renyi.hu/conferences/erdos100/program.html}

\bibitem[\protect\citeauthoryear{Heule \bgroup \em et al.\egroup
  }{2013}]{druptrim}
Marijn Heule, Warren A.~Hunt Jr., and Nathan Wetzler.
\newblock Trimming while checking clausal proofs.
\newblock In {\em Proceedings of Formal Methods in Computer-Aided Design, FMCAD
  2013}, pages 181--188. IEEE, 2013.

\bibitem[\protect\citeauthoryear{Konev and Lisitsa}{2014a}]{KLArx14}
Boris Konev and Alexei Lisitsa.
\newblock A SAT attack on the {Erd\H{o}s} discrepancy conjecture.
\newblock {\em CoRR}, abs/1402.2184, 2014.

\bibitem[\protect\citeauthoryear{Konev and Lisitsa}{2014b}]{KLSAT14}
Boris Konev and Alexei Lisitsa.
\newblock A SAT attack on the {Erd\H{o}s} discrepancy conjecture.
\newblock In {\em Proceedings of the 17th International Conference on Theory
  and Applications of Satisfiability Testing, SAT 2014}, 2014.
\newblock To appear.

\bibitem[\protect\citeauthoryear{Mathias}{1993}]{Mathias}
A.~R.~D. Mathias.
\newblock On a conjecture of {E}rd{\H{o}}s and {{\v{C}}}udakov.
\newblock Combinatorics, geometry and probability, 1993.

\bibitem[\protect\citeauthoryear{Matou\v{s}ek and Spencer}{1996}]{MS96}
Ji\v{r}\'{\i} Matou\v{s}ek and Joel Spencer.
\newblock Discrepancy in arithmetic progressions.
\newblock {\em Journal of the American Mathematical Society}, 9(1):195--204,
  1996.

\bibitem[\protect\citeauthoryear{Matou\v{s}ek}{1999}]{Matusek}
Ji\v{r}\'{\i} Matou\v{s}ek.
\newblock {\em Geometric Discrepancy: An Illustrated Guide}, volume~18 of {\em
  Algorithms and combinatorics}.
\newblock Springer, 1999.

\bibitem[\protect\citeauthoryear{Muthukrishnan and
  Nikolov}{2012}]{Muthukrishnan:2012:OPH:2213977.2214090}
S.~Muthukrishnan and Aleksandar Nikolov.
\newblock Optimal private halfspace counting via discrepancy.
\newblock In {\em Proceedings of the 44th Symposium on Theory of Computing},
  STOC '12, pages 1285--1292, New York, NY, USA, 2012. ACM.

\bibitem[\protect\citeauthoryear{Nikolov and Talwar}{2013}]{NiTa}
Aleksandar Nikolov and Kunal Talwar.
\newblock On the hereditary discrepancy of homogeneous arithmetic progressions.
\newblock {\em CoRR}, abs/1309.6034v1, 2013.

\bibitem[\protect\citeauthoryear{Polymath}{2010}]{Polymath}
D.H.J. Polymath.
\newblock Erd{\H{o}}s discrepancy problem: {P}olymath wiki, 2010.
\newblock
Retrieved April, 10 2014 from
  \url{http://michaelnielsen.org/polymath1/index.php?title=The_Erdős_discrepancy_problem}

\bibitem[\protect\citeauthoryear{Polymath}{2011}]{Polymath2}
D.H.J. Polymath.
\newblock Human proof that completely multiplicative sequences have discrepancy
  greater than 2, 2011.
\newblock
Retrieved April, 10 2014 from
  \url{http://michaelnielsen.org/polymath1/index.php?title=Human_proof_that_completely_multiplicative_sequences_have_discrepancy_at_least_2}

\bibitem[\protect\citeauthoryear{Rautenberg}{2010}]{rautenberg}
Wolfgang Rautenberg.
\newblock {\em A Concise Introduction to Mathematical Logic}.
\newblock Springer, New York, 3 edition, 2010.

\bibitem[\protect\citeauthoryear{Roth}{1964}]{Roth64}
Klaus~F. Roth.
\newblock Remark concerning integer sequence.
\newblock {\em Acta Arithmetica}, 9:257--260, 1964.

\bibitem[\protect\citeauthoryear{Roussel and Manquinho}{2009}]{cardconstrHB}
Olivier Roussel and Vasco~M. Manquinho.
\newblock Pseudo-boolean and cardinality constraints.
\newblock In Armin Biere, Marijn Heule, Hans van Maaren, and Toby Walsh,
  editors, {\em Handbook of Satisfiability}, volume 185 of {\em Fronteers in
  Artificial Intelligence and Applications}, pages 695--733. IOS Press,
  Amsterdam, 2009.

\bibitem[\protect\citeauthoryear{Sinz}{2005}]{Sinz05}
Carsten Sinz.
\newblock Towards an optimal {CNF} encoding of boolean cardinality constraints.
\newblock In {\em Principles and Practice of Constraint Programming - CP 2005,
  11th International Conference, CP 2005, Sitges, Spain, October 1-5, 2005,
  Proceedings}, volume 3709 of {\em Lecture Notes in Computer Science}, pages
  827--831. Springer, 2005.

\bibitem[\protect\citeauthoryear{Tseitin}{1970}]{Tseitin70}
Grigory~S. Tseitin.
\newblock On the complexity of derivation in propositional calculus.
\newblock In Anatol~O. Slisenko, editor, {\em Studies in Constructive
  Mathematics and Mathematical Logic, Part 2}, pages 115--125, New York, 1970.
  Consultants Bureau.

\bibitem[\protect\citeauthoryear{\v{C}udakov}{1956}]{chudakov}
Nikolai \v{C}udakov.
\newblock Theory of the characters of number semigroups.
\newblock {\em Journal of Indian Mathematical Society}, 20:11--15, 1956.

\end{thebibliography}
\end{document}